\documentclass[twocolumn, nofootinbib]{revtex4}

\usepackage{graphicx}

\usepackage{amsmath}
\usepackage{amsthm}

\theoremstyle{plain}

\theoremstyle{plain}
\newtheorem{lem}{Lemma}

\theoremstyle{definition}
\newtheorem{defn}{Definition}

\newtheorem{fact}{Fact}

\theoremstyle{remark}

\begin{document}

\title{A fault-tolerant non-Clifford gate for the surface code in two dimensions}

\author{Benjamin J. Brown}
\affiliation{Centre for Engineered Quantum Systems, School of Physics, University of Sydney, Sydney, New South Wales 2006, Australia}

\begin{abstract}
Fault-tolerant logic gates will consume a large proportion of the resources of a two-dimensional quantum computing architecture. Here we show how to perform a fault-tolerant non-Clifford gate with the surface code; a quantum error-correcting code now under intensive development. This alleviates the need for distillation or higher-dimensional components to complete a universal gate set. The operation uses both local transversal gates and code deformations over a time that scales with the size of the qubit array. An important component of the gate is a just-in-time decoder. These decoding algorithms allow us to draw upon the advantages of three-dimensional models using only a two-dimensional array of live qubits. Our gate is completed using parity checks of weight no greater than four. We therefore expect it to be amenable with near-future technology. As the gate circumvents the need for magic-state distillation, it may reduce the resource overhead of surface-code quantum computation considerably.
\end{abstract}

\maketitle

\section{Introduction}

A scalable quantum computer is expected solve difficult problems that are intractable with classical technology~\cite{Harrow17}. Scaling such a machine to a useful size will necessarily require fault-tolerant components that protect quantum information as the data is processed~\cite{Calderbank96, Steane96, Dennis02, Terhal15, Brown16, Campbell17}. If we are to see the realisation of a quantum computer, its design must respect the constraints of the quantum architecture that can be prepared in the laboratory. In many cases, for instance superconducting qubits~\cite{Raussendorf07, DiVincenzo09, Fowler12a, Corcoles15, Kelly15, Takita16}, this restricts us to two-dimensional architectures. 

Leading candidate models for fault-tolerant quantum computation are based on the surface code~\cite{Kitaev03, Dennis02} due to its high threshold~\cite{Dennis02, Wang03, Raussendorf07, Fowler12a} and multitude of ways of performing Clifford gates~\cite{Raussendorf06, Bombin09, Hastings15, Brown17, Yoder17}. Universal quantum computation is possible if this gate set is supplemented by a non-Clifford gate. Among the most feasible approaches to realise a non-Clifford gate is by the use of magic-state distillation~\cite{Bravyi05, Gidney19, Litinski19}. However, this is somewhat prohibitive as it is estimated that a large fraction of the resources of a quantum computer will be expended by these protocols~\cite{Fowler12a, OGorman17}.

Here we provide a promising alternative to magic-state distillation with the surface code. Remarkably, we show that we can perform a fault-tolerant non-Clifford gate with three overlapping copies of the surface code that interact locally. Each of the two-dimensional arrays of live qubits replicates a copy of the three-dimensional generalisation of the surface code over a time that scales with the size of the array. We use that the full three-dimensional model is natively capable of performing a controlled-controlled-phase gate~\cite{Kubica15, Vasmer18} to realise a two-dimensional non-Clifford gate. The procedure makes essential use of just-in-time gauge fixing; a concept recently introduced by Bomb\'{i}n in Ref.~\cite{Bombin18}. This enables us to recover the three-dimensional surface-code model using parity measurements of weight no greater than four. Research into such technology is presently under intensive development~\cite{Corcoles15, Kelly15, Takita16}, as these are the minimal requirements to realise the surface-code model.

The non-Clifford gate presented here circumvents fundamental limitations of two-dimensional models~\cite{Eastin09, Bravyi13a, Pastawski15, OConnor18, Webster18} by dynamically preparing a three-dimensional system using a two-dimensional array of active qubits. In the past, there has been a significant effort to realise a non-Clifford gate with two-dimensional quantum error-correcting codes~\cite{Bravyi15, OConnor16, Jones16, Yoder16, Yoder17a}. However, these proposals are unlikely to function reliably as the size of the system diverges. It is remarked in Ref.~\cite{Bombin18} that we should understand fault-tolerant quantum operations, not in terms of quantum error-correcting codes, but instead by the processes they perform. Notably, in our scheme, error-detecting measurements are realised dynamically. This is in contrast to the more conventional approach where we make stabilizer measurements on static quantum error-correcting codes to identify errors. As we will see, the process is well characterised by connecting the surface code with the topological cluster-state model~\cite{Raussendorf05, Raussendorf06, Raussendorf07, Raussendorf07a}, a measurement-based model with a finite threshold error rate below which it will function reliably at a suitably large system size. Furthermore, as we will see, the cluster state offers a natural static language to characterise the dynamical quantum process using a time-independent entangled resource state.

We begin by defining measurement-based model, and we explain how we project the non-Clifford gate onto a two-dimensional surface. We finally discuss the just-in-time decoder that permits a two-dimensional implementation of the gate. Microscopic details of the system and proof of its threshold are deferred to Methods.

\section{Results}

\subsection*{The topological cluster state}

The topological cluster-state model~\cite{Raussendorf05} is described in three dimensions. However, we need only maintain a two-dimensional array of its qubits at a given moment to realise the system~\cite{Raussendorf07}. Specifically, we destructively measure each qubit immediately after it has interacted with all of the other qubits that are specified by the cluster state. This method of generating the model on the fly gives rise to a time-like direction, see Fig.~\ref{Fig:Raussendorf}(a).

We use the topological cluster state to realise the three-dimensional surface code~\cite{Hamma05}. We define the surface code on a lattice with arbitrary geometry with one qubit on each edge, $e$. The model is specified by two types of stabilizers, star and plaquette operators, denoted $A_v$ and $B_f$, see Methods~\ref{SubSec:Lattices} for details. Stabilizers specify the code states of the model, $|\psi \rangle$, such that $A_v |\psi \rangle = B_f |\psi \rangle = |\psi \rangle$ for all code states. Star operators are associated to the vertices, $v$, of the lattice such that $A_v = \prod_{\partial e \ni v} X_e$ where $\partial e$ is the set of vertices at the boundary of $e$ and $X_e $ and $Z_e$ are Pauli operators acting on $e$. Plaquettes $B_f$ lie on lattice faces $f$ such that $B_f = \prod_{e \in \partial f} Z_f$ where $\partial f$ are the edges that bound $f$. We give details on explicit lattice geometries that we might use in Methods~\ref{SubSec:Lattices}.

\begin{figure}
\includegraphics{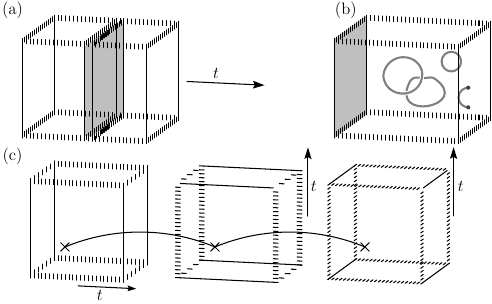}
\caption{(a)~The topological cluster-state model is a three-dimensional model that propagates quantum information over time with only a two-dimensional array of live qubits at any given moment. We show a grey plane of live qubits that propagates in the direction of the time arrow. (b)~Grey loops of show the connectivity of plaquette measurements that returned the -1 outcome. An arbitrary state is initialised fault tolerantly by initialising the system with an encoded two-dimensional fixed-gauge surface code on the grey face at the left of the image. (c)~The boundary configurations of the three copies of the surface code required to perform a local transversal controlled-controlled-phase gate. The first code has rough boundaries on the top and the bottom of the lattice. The middle(right) code has rough boundaries on the left and right(front and back) sides of the lattice. The orientation of the boundaries determines the time direction in which we can move the planes of live qubits. \label{Fig:Raussendorf}}
\end{figure}

To connect the three-dimensional surface code with the topological cluster state~\cite{Raussendorf05} we consider initialising the surface code in the $+1$ eigenvalue eigenstate of the logical Pauli-X operator by measurement. We consider initialising all of the physical qubits in the $|+\rangle_e$ state and then measure all of the plaquette operators. Up to an error correction step, this completes initialisation. To measure a plaquette operator $B_f$ we prepare an ancilla qubit, $a$, in the $| + \rangle _a$ state and couple it to the qubits that bound $f$ with controlled-phase gates, i.e., we apply $U = \prod_{e \in \partial f} CZ_{e, a}$ with $CZ_{j,k} = (1 + Z_j + Z_k - Z_jZ_k) / 2$ the controlled-phase gate. It may be helpful to imagine placing the ancilla qubit on face $f$. Measuring the ancilla qubit in the Pauli-X basis will recover the value of the face operator. However, we observe that before the ancillas are measured we have the topological cluster-state model~\cite{Raussendorf05} where now the qubits of the surface code give the qubits of the primal lattice of the model and the ancilla qubits make up the qubits of its dual lattice.


The require that the surface code lies in the $+1$ eigenvalue eigenstate of its face operators. However, measuring all the dual qubits of the cluster state projects its primal qubits into a random gauge of the three-dimensional surface code where, up to certain constraints, all the face measurements take random values. Henceforth, unless there is ambiguity, we refer to the model with face operators fixed onto their $+1$ eigenvalue eigenstate as the surface code. Otherwise we call it the random-gauge surface code. It is important to realise the fixed gauge surface code to perform the controlled-controlled-phase gate~\cite{Vasmer18}.

We use error correction to recover the surface code from the random gauge model~\cite{Paetznick13, Anderson14, Bombin15, Vuillot18}. We note that the product of all the face operators that bound a cell return identity, i.e. $\prod_{f \in \partial c} B_f = 1$ where $\partial c$ at the set of faces that bound cell $c$. As such, supposing all of the measurements are made noiselessly, there must be an even parity of measurements that return the $-1$ outcome about each cell. This in turn constrains the plaquette operator measurements to respect loop-like configurations on the dual lattice, see Fig.~\ref{Fig:Raussendorf}(b). To recover the fixed-gauge surface code we apply a Pauli-X operator with a membrane like support whose boundary terminates at each component of the loop configuration. 

Further, we can initialise the surface code in an arbitrary state fault tolerantly if, before face measurements are made, we replace the unentangled qubits on one side of one boundary of the lattice with an encoded surface code, for instance the grey face shown to the left of Fig.~\ref{Fig:Raussendorf}(a)~\cite{Brown18}. We refer to this face as the initial face. Imposing that the face operators of the surface code are fixed in the $+1$ eigenvalue eigenstate mean no loop configurations will terminate at this boundary. This method of initialisation is a dimension jump~\cite{Bombin16}.

\subsection*{Embedding the non-Clifford gate in two dimensions}

We can now explain how we can embed the three-dimensional surface code that performs a non-Clifford gate in two dimensions. There are several constraints the system must satisfy if we realise a controlled-controlled-phase gate with a two-dimensional system. We first point out that the orientation of boundaries of the topological cluster state are important for the transmission of logical information~\cite{Brown18}. Moreover, they constrain the temporal directions of the model. We consider again the cluster state in terms of the three-dimensional surface code. The surface code model has two types of boundary; rough and smooth~\cite{Hamma05}. If we couple ancilla to the surface code to recover the topological cluster state as specified above then the rough(smooth) boundaries of the surface code give rise to the primal(dual) boundaries of the cluster state~\cite{Raussendorf06}. If we only maintain a two-dimensional array of qubits, the plane must contain two distinct primal boundaries that are well separated by two distinct dual boundaries to support the encoded information. The grey plane in Fig.~\ref{Fig:Raussendorf}(a) is suitable, for example.

Secondly, the boundaries of the three surface codes must be correctly configured to perform a transversal controlled-controlled-phase gate~\cite{Kubica15, Vasmer18}. Fig.~\ref{Fig:Raussendorf}(c) shows the boundaries configured such that the qubit at coordinate $P = (x,y,z)$ of each code interacts with the respective qubit at the same location of the other codes via transversal controlled-controlled-phase gates. To perform the gate locally, these three lattices must overlap while maintaining these boundary conditions.

Finally, if we only maintain a two-dimensional array of the three-dimensional system, it is important that all of the qubits that need to interact with one another must be live at the same time. We show a system that satisfies all of these constraints in Fig.~\ref{Fig:Locality}. The figure shows a three-dimensional spacetime diagram of two overlapping codes moving orthogonally to one another. We omit the third code as it can travel in parallel to one of the codes already shown. The first code that has rough boundaries on the top and bottom of its volume and the live plane of qubits moves right across the page. The second code has rough edges on its left and right faces and moves upwards in the Figure. The controlled-controlled-phase gate is made at the cubic region where the codes intersect. We find that all of the appropriate qubits are active at the right moment by choosing two diagonal planes of live qubits for each code. We can also see that the planes we choose all have two well separated rough and smooth boundaries within their respective volume.

\begin{figure}
\includegraphics{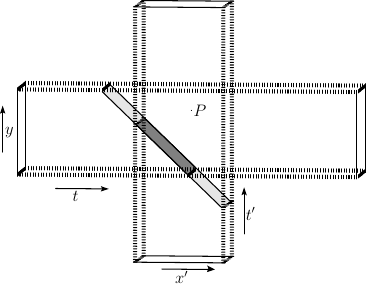}
\caption{Two codes travelling in different temporal directions cross. The third code is omitted as it can run in parallel with one of the two shown. Live qubits of the spacetime history are shown on light grey planes. The transversal gate is applied in the cubic region in the middle. It will be applied on the qubits shown at the dark grey plane where the two-dimensional arrays of qubits are overlapping. \label{Fig:Locality}}
\end{figure}

We are now ready to consider an embedding the three-dimensional spacetime shown in Fig.~\ref{Fig:Locality} onto a two-dimensional manifold. We find that one of the codes has to move with respect to the other. This can be naturally incorporated in the procedure to generate layers of the topological cluster state, see Methods~\ref{SubSec:GaugeFix}. We consider a point $P$ in the spacetime diagram in the region where the controlled-controlled-phase gate is performed such that a qubit of each of the two models must interact. The coordinates of the locations of the two codes change differently with time. The first code that travels upwards in the spacetime diagram has coordinates $P = (x' , t')$, the other that moves from left to right has coordinates $P= (t, y)$ with time $t = t'$. We neglect the $z$-coordinate as this is static. We imagine projecting the three-dimensional system onto a two-dimensional plane such that $y = t' = 0$, we now observe that $t = x'$. We conclude that one code must move with respect to the other two static codes to ensure all of the qubits that must interact are local at the right points in time.

\subsection{Just-in-time gauge fixing}

We use a decoder to fix the topological cluster-state model onto the surface code using data from the ancilla qubit measurements. In the case that there are measurement errors we will necessarily introduce small Pauli errors onto the system that will translate into Clifford errors upon application of the transversal non-Clifford gate. Measurement errors in this model take the form of strings which are detected by defects that lie at their endpoints. A decoder must attempt to close these endpoints. We can pair the defects with conventional decoders for topological codes such as minimum-weight perfect matching~\cite{Dennis02} or clustering~\cite{Bravyi13b}. We then fix the gauge by applying a membrane-like Pauli-X operator whose boundary is the union of both the error and the correction.

Small discrepancies in the correction compared with the actual measurement error will lead to gauge-fixing errors. Then, applying the transversal operation to the system whose gauge has been fixed incorrectly will introduce controlled-phase gates between pairs of qubits of different surface codes that have been involved in the same controlled-controlled-phase interaction. Provided the errors that are introduced during gauge fixing are small and are supported on a correctable region, the errors the transversal gate will introduce are also correctable. When we infer the values of the star operators we project the Clifford errors that are diagonal in the computational basis onto Pauli-Z errors. After the projection these errors also manifest as strings on the three-dimensional surface code. Again, we detect the string-like errors by measuring defects at the endpoints of the strings. Once more, we can correct these errors with any suitable algorithm that pairs the defects. We can therefore prove a fault-tolerance threshold under the gauge fixing procedure by showing the errors we introduce during gauge fixing are small in comparison to the distance of the code.

\begin{figure}
\includegraphics{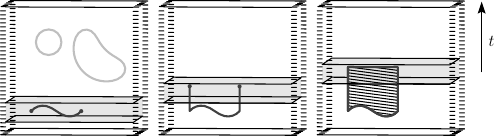}
\caption{\label{Fig:JIT} The spacetime diagram of an error on the dual qubits of the topological cluster-state where time travels upwards. The grey area shows the two-dimensional area of live qubits at a given moment. At the point where an error is discovered on the left diagram, it is unlikely that the defects should be paired due to their separation. We therefore defer matching the defects to a later time after more information emerges as decoding progresses, as in the middle figure. After enough time, the most likely outcome is that the defects we discovered in the left figure should be paired. The error we introduce fills the interior of the error, and the chosen correction.}
\end{figure}

We aim to fix the gauge of a three-dimensional model. However, we will only maintain a two-dimensional array of live qubits. As such the decoder has a limited amount of information available to make decisions about how to pair defects. To overcome this issue we defer correcting pairs of defects to a later time once we have more certainty that two defects should be matched. This leads the errors to spread over the time matching is deferred, see Fig.~\ref{Fig:JIT}. We propose a renormalisation-group~\cite{Harrington04, Bravyi13b} just-in-time decoder~\cite{Bombin18} that will defer the pairing of defects such that the spread of errors is controlled. Broadly speaking, we find that a just-in-time decoder will work if the pairing of two defects is deferred until both defects have existed for a time proportional to their separation in spacetime. We make this statement precise in Methods~\ref{App:JIT} and prove it controls the spread of errors. Moreover, we discuss how the decoder is modified to find a correction close to the boundary of the surface code.

Supposing an independent and identically distributed error model that is characterised in terms of chunks~\cite{Gacs86, Grimmett99, Gray01} we can show that the just-in-time decoder will not spread a connected component of the error by more than a constant factor of the size of the component. We further find that this spread error model can be decoded by a renormalisation-group decoder. We prove a threshold against the spread error model using a renormalisation-group decoder in Methods~\ref{SubSec:ThresholdTheorem} , see Lemma~\ref{Lemma:RG}. We then prove that the just-in-time decoder will give rise to spread errors with a constant spread, see Lemma~\ref{Lemma:JIT} in Methods~\ref{SubSec:JustInTime}, thus justifying the noise model. In contrast, the threshold theorem for just-in-time decoding given by Bomb\'{i}n.~\cite{Bombin18} uses a minimum-weight perfect matching decoder~\cite{Edmonds65, Dennis02}.

\section{Error mitigation}

One should worry that the just-in-time gauge fixing process will add errors that may significantly decrease the logical failure rate of the system. We argue that we can make this effect relatively benign in post processing. The errors introduced by the just-in-time decoder are twofold. Firstly, it may directly introduce a logical failure by incorrectly matching defects and, secondly, if the decoder does succeed, it will introduce large errors to the primal qubits of the system that need to be decoded globally once the gate is complete. 

Rather than considering the protocol as a gate that can be used on the fly in some algorithm, we use it to produce high-fidelity magic states by inputting Pauli-X eigenstates that are prepared fault tolerantly. Once gauge fixing is completed, we can simulate gauge-fixing again with a high-performance global decoder~\cite{Wang03, Raussendorf06, Wootton12}. We can then compare the output of the high-performance decoder with the just-in-time decoder. If their results do not agree, we discard the output. 

We denote the failure rates of the high-threshold(just-in-time) decoder $\overline{P}_{\text{HP}}$($\overline{P}_{\text{JIT}}$). Both decay rapidly with system size below threshold, but we suppose $\overline{P}_{\text{HP}} \ll \overline{P}_{\text{JIT}}$. In the event that the decoders disagree, we discard the state. This occurs with likelihood $ \sim \overline{P}_{\text{JIT}}$. In the case that the decoders agree, the state that we output is logically incorrect with likelihood $\sim \overline{P}_{\text{HP}}\overline{P}_{\text{JIT}}$. The use of a high-threshold decoder therefore improves the fidelity of the post-selected output states. The failure rate of the just-in-time decoder then only determines the rate at which magic states should be discarded. 

We can also use the output of the high-performance decoder to deal with the errors spread to the surface code with just-in-time decoding. We can compare the output of the high-performance decoder with the correction produced by the just-in-time decoder. The discrepancy in their outputs should indicate the approximate locations of the gauge-fixing errors. This information can be fed to the decoder we use to decode the errors on the primal qubits that will flag the discrepancy as qubits that are highly likely to support an error. Indeed there have been a number of results showing how to improve decoders by using knowledge of likely errors~\cite{Delfosse14a, Fowler13, Nickerson17, Criger18}. Ref.~\cite{Bombin18} treats these flagged qubits as erasure errors that support linking charges~\cite{Yoshida15, Bombin18a}. The proof given in Methods~\ref{App:JIT} shows that we have a threshold without these considerations, but implementations of this protocol should use a decoder that accounts for these effects to improve their performance. After post-selection then we might expect the system to perform as though it were gauge fixed globally with some known erasure errors~\cite{Barrett10}.

\section{Discussion}

To summarise, we have shown how to perform a fault-tolerant controlled-controlled-phase gate with a two-dimensional surface-code architecture and we have proved it has a threshold. Next, it is important to compare the resource scaling of this scheme compared with more conventional two-dimensional approaches to fault-tolerant quantum computation, namely, surface-code quantum computation with magic state distillation~\cite{Bravyi05, Fowler12a}. Given that gauge-fixing errors will spread phase errors as we apply the three-qubit transversal gate, the logical error rate of this scheme is likely to decay more slowly than approaches using magic state distillation where we do not rely on gauge fixing. However, the spacetime volume of realising a fault-tolerant controlled-controlled-phase gate, $\sim 30 d^3$, see Methods~\ref{SubSec:Protocol}, is an order of magnitude smaller than a single distillation routine, as such, these schemes are clearly deserved of further comparison. It is likely that the optimal choice will depend on the error rate of the physical hardware.

It will also be interesting to compare the protocol introduced here to that presented by Bomb\'{i}n~\cite{Bombin18}. This two-dimensional non-Clifford gate is based on the color code such that a transversal $T = \text{diag}(1, \text{i}^{1/2}) $ gate is performed over time via single-qubit rotations. This will make for a very interesting comparison since, even though decoding technology for the color-code model~\cite{Delfosse14, Bombin15a, Brown16a, Kubica17, Kubica18, Aloshious18, Brown19} remains lacking in comparison to the surface code~\cite{Raussendorf07, Fowler12a}, the fact that the non-Clifford operation is performed using single-qubit rotations instead of a weight-three gate will mean that fewer errors will be spread during computational processes. To begin comparing these protocols fairly it will first be important to improve the decoding algorithms we have for the color code.

Finally, it is likely that there will be several ways to optimise the present scheme. Although we find transversal gates via a mapping between the color code and the surface code~\cite{Kubica15, Vasmer18} such that we arrive at quite a specific lattice, it will be surprising if we cannot find ways of performing a constant-depth locality-preserving gate with other lattices~\cite{Nickerson18} for the topological cluster-state model. Indeed, history has shown that the gates a given model is able to achieve is connected with the macroscopic properties of a system, not its microscopic details. Developing our understanding of measurement-based quantum computation by decomposing it in terms of its topological degrees of freedom~\cite{Bombin18, Brown18} is likely to be a promising route towards better models of two-dimensional fault-tolerant quantum computation.

\begin{acknowledgements}
We are grateful to A. Doherty, M. Kesselring, N. Nickerson and S. Roberts for helpful and supportive discussions, and in particular S. Bartlett, C. Chubb, C. Dawson and S. Flammia for patiently listening to various incarnations of these results. We thank S. Bartlett and S. Flammia for critically reading earlier drafts of this manuscript, and we thank D. Poulin for many discussions on non-Clifford operations. We are also grateful to H. Bomb\'{i}n for correspondence on gauge prefixing with the color code. This work is supported by the University of Sydney Fellowship Programme and the Australian Research Council via the Centre of Excellence in Engineered Quantum Systems(EQUS) project number CE170100009. The author declares that he has no competing interests. All information needed to evaluate the conclusions of the paper are present in the paper. Additional information related to this paper may be requested from the author.
\end{acknowledgements}

\section*{Methods}

\subsection{Lattices and mobile qubits}
\label{App:Microscopics}

Here we describe the microscopic details and dynamics of the system. We describe the lattice, and how the gauge-fixing progresses. We finally discuss the protocol over its entire duration to estimate its resource cost.

\subsubsection{Lattices}
\label{SubSec:Lattices}

\begin{figure}[b]
\includegraphics{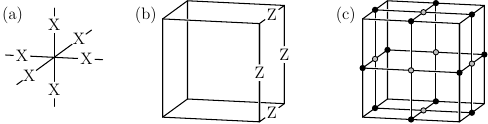}
\caption{The star, (a), and plaquette, (b), of the surface code where qubits lie on the edges of a cubic lattice. (c)~A unit cell of the topological cluster state. Black qubits are those of the three-dimensional surface code. Measuring the grey ancilla recover the values of the face operators of the surface code whose qubits lie on the edges of a cubic lattice. \label{Fig:Stabilizers}}
\end{figure}

In Ref.~\cite{Vasmer18} the authors describe three surface codes on different three-dimensional lattices. We give simple representations of the lattices here that help understand the steps of gauge fixing. The first of the three copies is well represented with the standard convention where qubits lie on the edges of a cubic lattice. We show star and plaquette operators in Fig.~\ref{Fig:Stabilizers}(a) and~(b). The other two lattices are represented with qubits on the vertices of rhombic dodecahedra in Ref.~\cite{Vasmer18}. We offer an alternative description of this model in this Subsection.

\begin{figure}
\includegraphics{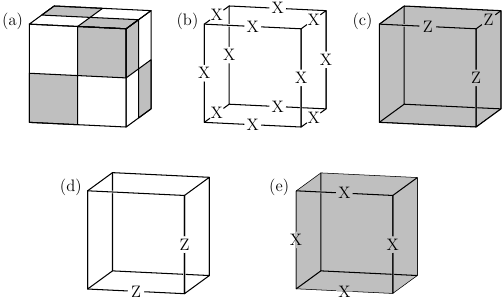}
\caption{\label{Fig:AltLattice} A representation of one of the three-dimensional surface codes used in the controlled-controlled-phase gate. (a)~A unit cell is composed of four primal cubes and four dual cubes configured as shown with primal and dual cubes shown in white and grey, respectively. (b)~A Pauli-X `star' operator supported on a primal cube. (c)~A `plaquette' operator supported on the corner of a dual cube. (d)~A smooth boundary stabilizer. (e)~A rough boundary stabilizer.}
\end{figure}

All of the qubits are unified with the qubits of the first surface code on the cubic lattice. We therefore find a straight forward way of representing the stabilizers of the second model with qubits on the edges of a cubic lattice. We show the stabilizers in Fig.~\ref{Fig:AltLattice} on a cubic lattice. To represent this model we bicolour the cubes, as they support different stabilizers depending on their colour, see Fig.~\ref{Fig:AltLattice}(a). The white primal cubes support Pauli-X `star' operators and the grey dual cubes support the Pauli-Z `plaquette' operators. We express their support with the following equations
\begin{equation}
A_c = \prod_{e\in \partial c} X_e, \quad B_{c,v} = \prod_{\substack{  \partial e \ni v  \\ e \in \partial c}} Z_e, \label{Eqn:AltStabilizers}
\end{equation}
where $\partial c$ are the set of edges on the boundary of cube $c$ and again, $\partial e$ is the set of vertices $v$ at the boundaries of edge $e$, i.e., its endpoints. The operators $A_c$ and $B_{c,v}$ are, respectively, defined on primal and dual cubes only. We also note that each vertex touches four dual cubes, as such there are four $B_{c,v}$ at each vertex. Further, there are eight vertices on a cube, there are therefore eight $B_{c,v}$ stabilizers for each dual cube $c$. There is only one $A_c$ operator for each primal cube. We also show the stabilizers added at the smooth- and rough-boundaries in Figs.~\ref{Fig:AltLattice}(d) and~(e) respectively. See Ref.~\cite{Vasmer18} for a more detailed discussion on the boundaries.

For convenience, we show the same lattice in the standard convention where star and plaquette operators are associated to the vertices and faces, respectively; see Fig.~\ref{Fig:AltStarPlaquettes}. In Fig.~\ref{Fig:AltStarPlaquettes}(a) we show two cells; one with four triangular faces and one with eight. The product of the faces about either of these cells are constrained to give identity. These constraints are important for the gauge-fixing procedure.

It is helpful to connect the face operators of the different cells shown in Fig.~\ref{Fig:AltStarPlaquettes}(a) to their respective plaquette operators $B_{c,v}$ as represented in Fig.~\ref{Fig:AltLattice} and by Eqn.~(\ref{Eqn:AltStabilizers}). Indeed, the four plaquette operators that enclose the four-sided cell to the left of the image are the four stabilizers $B_{c,v}$ about a common vertex, $v$, as represented in Fig.~\ref{Fig:AltLattice}. Likewise, the eight-sided cell at the right of the image supports all the stabilizers $B_{c,v}$ about a common cube, $c$. 

\begin{figure}
\includegraphics{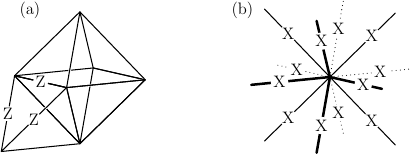}
\caption{(a)~A plaquette and~(b) a star operator of the alternative lattice with the qubits on the edges of a lattice and stars(plaquettes) are associated to the vertices (faces) of the lattice to keep with the standard convention for describing the surface code. In~(b) edges in the foreground are bold, and edges in the background are dotted. \label{Fig:AltStarPlaquettes}}
\end{figure}

As discussed in the main text, the cells of the surface code are connected to cells of the topological cluster-state model. In Fig.~\ref{Fig:Stabilizers}(c) we show a cell of the topological cluster-state model that is obtained by coupling ancilla qubits $f$ prepared in the state $|+\rangle_f$, shown in grey in the figure, to the faces $\partial f$ with a controlled-phase gate. We  measure the $B_f$ operator if the ancilla is measured in the Pauli-X basis. The ancilla qubits in this picture are the qubits of the dual lattice of the topological cluster state. The controlled-phase gates used to couple the ancilla to the black qubits of the surface code are shown by black edges.

Finally, we count the number of qubits in a single unit cell, as these will make up a site in the threshold theorem given in Methods~\ref{App:JIT}. There are three qubits per cube, so over a unit cell of eight cubes we have 24 qubits. We also include ancilla qubits to measure the plaquette operators of each model. In the cubic lattice model we make one plaquette measurement for each face of the lattice. There are three faces per cube of the lattice, we therefore have 24 ancilla qubits to measure the faces of the cubic lattice model. For the alternative lattice we make eight measurements per dual cube of the unit cell. We have four dual cubes per unit cell, we therefore arrive at 32 ancilla qubits for each unit cell of the alternative lattice shown in Fig.~\ref{Fig:AltLattice}. We therefore have 48 qubits in total per unit cell of the cubic lattice model, and 56 qubits per unit cell of the alternative model.

We now consider a unit cell of the total system with three overlapping lattices. Each unit cell includes one copy of the cubic lattice model and two copies of the alternative model. We therefore find that we have 160 qubits per unit cell in total. The unit cells at the boundary of the system can be regarded as bulk cells with some of the qubits removed. As such, when we accound for the boundary, we can take this value as an upper bound. Lastly, we note that each of these unit cells contributes two units of distance to the system.

\subsubsection{Gauge-fixing}
\label{SubSec:GaugeFix}

Having specified the lattices, we now discuss how to perform the gauge-fixing process. Gauge fixing moves three two-dimensional surface codes through a three-dimensional spacetime volume to reproduce three overlapping three-dimensional surface codes over time. This motion proceeds by repeatedly producing a thin layer of three-dimensional surface code and then measuring some of its qubits in a product basis to collapse the system onto a two-dimensional surface code that has been displaced through spacetime. Gauge fixing, and transversal controlled-controlled-phase gates are applied at the intermediate step where the system is in the state of a thin slice of three-dimensional surface code. We show one period of the process for two lattices in Fig.~\ref{Fig:LatticeSurgeryPeriod}. Each panel of the figure shows the region in which the transversal controlled-controlled-phase gate is conducted within the black cube. The top figures show the progression of a lattice moving from left to right through the region over time, and the lower figures show a lattice moving upwards through the region. Time progresses from left to right through the panels. The columns of the diagram are synchronised.

We now describe the microscopic details of a single period of the gauge fixing process. We perform similar processes on all three surface codes involved in the gate in unison. The three surface codes only differ in the direction they move through the spacetime volume, and the lattice we use to realise the surface code. As such we will only focus on a single surface code, say that shown in Fig.~\ref{Fig:LatticeSurgeryPeriod}(a). 

A period of the gauge-fixing process begins with a two-dimensional surface code supported on the qubits shown at time $t$ to the left of Fig.~\ref{Fig:LatticeSurgeryPeriod}(a) and it ends at time $t+1$ with a displaced surface code, shown in the right column of the figure. It is helpful to label the subsets of qubits of the spacetime volume that support a surface code at time $t$($t+1$) with the label $\mathcal{Q}_t$($\mathcal{Q}_{t+1}$). The thin three dimensional-surface code that we produce at the intermediate step is shown in the central column of Fig.~\ref{Fig:LatticeSurgeryPeriod}(a) at time $t+1/2$. We denote the qubits that support the three-dimensional surface code at this time by $\mathcal{Q}_{t+1/2}$. The subsets of qubits we have defined are such that $\mathcal{Q}_{t},\, \mathcal{Q}_{t+1} \subset \mathcal{Q}_{t+1/2}$ and the intersection of $\mathcal{Q}_t$ and $\mathcal{Q}_{t+1}$ is non empty.

We map the surface code at time $t$ onto the three-dimensional surface code shown at time $t+1/2$ by measurement. We initialise the qubits in the subset $\mathcal{Q}_{t+1/2} \backslash \mathcal{Q}_t $ in the $|+\rangle$ state. We then measure all of the plaquettes supported on $\mathcal{Q}_{t+1/2}$ that have not been measured previously. Plaquettes supported entirely on $\mathcal{Q}_{t}$ have already been measured at an earlier period. It is therefore unnecessary to measure these stabilizers again.

The plaquette measurements will return random outcomes, and may include errors. We must fix the gauge of the plaquettes of the active layer of the surface code to their $+1$ eigenstate. This is described in more detail in Methods~\ref{App:JIT}. For now we assume that it is possible to accomplish this. Once we make the gauge-fixing correction, we apply the controlled-controlled-phase gate between the qubits of subset $\mathcal{Q}_{t+1/2} \backslash \mathcal{Q}_{t+1}$ of each of the three systems involved in the gate.

We finally recover a two-dimensional surface code on the subset of qubits $\mathcal{Q}_{t+1}$ by measuring the qubits of the subset  $\mathcal{Q}_{t+1/2} \backslash \mathcal{Q}_{t+1}$ in the Pauli-X basis. We use the outcomes of the destructive single-qubit Pauli-X measurements to infer the values of the star operators of the three-dimensional surface code. As measurement errors that occur when we make single-qubit measurements are indistinguishable from physical errors, the readout of the star operators of the three-dimensional surface code is fault tolerant. 

 In a sense, we can consider this as a dimension jump~\cite{Bombin16} where a two-dimensional model is incorporated into a three-dimensional model to leverage some property of the higher-dimensional system. In this case, we prepare a very thin slice of the three-dimensional surface code model where, once all the physical operations have been performed, we can collapse the three-dimensional model back onto a two-dimensional model again. As a remark, we point out that the latter dimensional jump where we go from the three-dimensional surface code to its two-dimensional counterpart has been accomplished by Raussendorf, Bravyi and Harrington~\cite{Raussendorf05} where they fault-tolerantly prepare a Bell pair between two surface codes using the topological cluster state.

\begin{figure}
\includegraphics{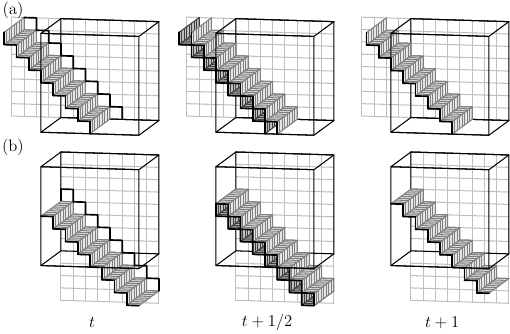}
\caption{One period of the gauge-fixing process for the models undergoing the controlled-controlled-phase gate. Time progresses between the figures from the left to the right from time $t$ to $t+1$ via an intermediate step at time $t+1/2$. (a)~The lattice above shows the code moving from left to right through the controlled-controlled-phase gate region, marked by the black cube, and (b)~the lower figures show the a code moving upwards through the black cubic region. Importantly, all of the overlapping qubits are live at the same time. The figures at the left show a two-dimensional surface code. In the middle figures, ancillas are added and the plaquette measurements of each of the cubes are made. The a gauge-fixing correction is made before the transversal controlled-controlled-phase gate is applied to each of the qubits on the lower levels that are subsequently measured destructively. \label{Fig:LatticeSurgeryPeriod}}
\end{figure}

Given that we have shown that all of the plaquettes can be measured about cubes of the lattice for both of the surface code geometries, it is enough just to consider the creation and collapse of cubes of the model, where the creation of a cube involves measuring all of the plaquettes within a cube. This is shown in Fig.~\ref{Fig:LatticeSurgeryPeriod}. At the left of the figure, the live qubits are two-dimensional surface codes in known eigenstates of their star and plaquette stabilizers. We initialise new ancilla close to the surface code and measure the plaquettes about each cube to produce the thin layer of the three-dimensional surface code. This is shown in the middle of the figure. We will collapse some of the qubits to produce the translated surface code shown to the right of the figure. Before this, we use the measurement data from the plaquettes to gauge fix the qubits using the just-in-time decoder and perform a controlled-controlled-phase gate on the qubits that are about to be measured. The outcomes of the single-qubit Pauli-X measurements are collected to infer the values of the star operators of the surface code model with a global decoder after the process is completed.

It is worth remarking that the method we have discussed here enables us to produce other three-dimensional structures that go beyond foliation~\cite{Nickerson18}. Much research has sought to map quantum error-correcting codes into measurement-based schemes~\cite{Bolt16, Brown18} through a system called `foliation' to access favourable properties of exotic quantum error-correcting codes. Conversely, some fault-tolerant measurement-based schemes have been developed that are not expected to have a description in terms of a quantum error-correcting code. Really though, at least in theory, we should expect that we can implement any fault-tolerant protocol independent of the architecture that we choose to realise our qubits. The scheme presented here gives us a way to realise these models that are beyond foliation with a two-dimensional array of static qubits. Given their promising thresholds~\cite{Nickerson18} it may be worth exploring the practicality of some of these higher-dimensional models on two-dimensional architectures.

In the same vein, we point out that the two-dimensional surface code that is propagated by the code deformations of the alternative lattice is the described naturally on the hexagonal lattice. This lattice has been largely dismissed because of the weight-six hexagonal stabilizer terms diminish the threshold against bit-flip noise~\cite{Delfosse14a}. However, we measure its stabilizers using only weight-three measurements, and the higher weight stabilizers are inferred from single-qubit measurements. As such, it may be worth revisiting this model as the scheme presented here offers a method of stabilizer extraction that does not require measurements of weight greater than three. We may therefore expect this model to have a high threshold with respect to the gate error model~\cite{Raussendorf06}.

\subsection{Implementing the non-Clifford gate}
\label{SubSec:Protocol}

We finally describe the entire protocol which is summarised in Fig.~\ref{Fig:Protocol}, and discuss its spacetime resource cost as a function of the code distance of the system, $d$. Each panel of the figure shows three arrays, each of which supports a code. It may be possible to embed the qubits of all three codes on one common array, but for visualisation purposes we imagine three stacked arrays that can perform local controlled-controlled-phase gates between nearby qubits on separate arrays. Parity measurements are performed locally on each array.

The code on the lower array will move from left to right along the page as we undergo code deformations. For a strictly local system we consider an extended array that we refer to as the long array. However, as we discuss towards the end of this section, we can reduce the size of this array by simulating a system with periodic boundary conditions. We proceed with the discussion where the process is strictly local. To evaluate the resource cost we refer to a single unit of time as a cycle. The resource cost is measured in units of qubit cycles.

\begin{figure}
\includegraphics{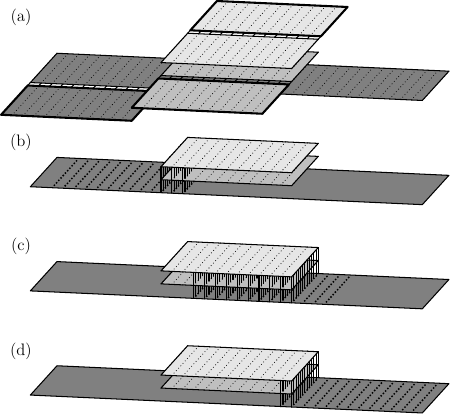}
\caption{{\bf A two-dimensional layout for the non-Clifford gate.} The progression of the controlled-controlled-phase gate. (a)~Qubits are copied onto the stacked chips from other surface codes using lattice surgery. (b)~The thick black qubits are passed under the other two chips and controlled-controlled phase gates are applied transversally between the three chips where the qubits overlap. (c)~and (d)~show later stages in the dynamics of the gate. \label{Fig:Protocol}}
\end{figure}

Before the gate begins we must copy the encoded information onto the arrays where the gate is performed. We might accomplish this with lattice surgery~\cite{Horsman12, Brown17}. Fig.~\ref{Fig:Protocol}(a) shows three surface codes that have been moved close to the edges of the arrays where the gate will be performed. One logical qubit is copied to the far left of the long array. Initialising the system will take time that scales like the code distance, $\sim d $ cycles.

We might also consider using the system offline to prepare high-fidelity magic states. With this setup we apply the gate to three surface codes initialised fault tolerantly in an eigenstate of the Pauli-X operator. While this will mean that we do not need to copy information onto the three arrays, it will still be necessary to fix the gauge of the system such that all the plaquette operators of the initial face are in their $+1$ eigenvalue eigenstate. To the best of our knowledge this will still take $\mathcal{O}(d)$ time to prepare the system such that its global charge is vacuum.

We remark that using the protocol offline to produce magic states may offer some advantages. For instance, as we discussed in the main text, we can post-select high-quality output states by comparing the result of the just-in-time decoder with a high-performance decoding algorithm. Moreover, the required connectivity of the gate with the rest of the system will be reduced. This is because we need only copy the magic states out of the system and we do not need to input arbitrary states into the system that may require additional routing.

Once the system is initialised we begin performing the code deformations as discussed in the previous section. The code deformations move the code on the long array underneath the other two codes, see Fig.~\ref{Fig:Protocol}(b), and out the other side, see Fig.~\ref{Fig:Protocol}(c). Assuming that one step as shown in Fig.~\ref{Fig:LatticeSurgeryPeriod} takes one cycle, moving the lower code all the way under the other two and out the other side will take $2d$ units of time. The final state of the protocol is shown in Fig.~\ref{Fig:Protocol}(d).

The above discussion explains that the three arrays will be occupied for $3d$ cycles. Each array will support a code that will consist of $\sim d \times d$ unit cubes that collectively can produce a thin slice of the three-dimensional surface code. Arrays of unit cubes are shown in Fig.~\ref{Fig:LatticeSurgeryPeriod} at time $t + 1/2$. The long array must be able to support unit cubes in $3d \times d$ locations. We include the idle qubits of the long array in the resource cost over the entire protocol. We count the qubits of each unit cube we need to realise each of the three-dimensional surface codes, including an ancilla qubit for each plaquette measurement we make on a given unit cube. We note that we have chosen the term `unit cube' here, as distinct from the `unit cell' that was defined in Methods~\ref{SubSec:Lattices}. The unit cell is a single element of a translationally invariant lattice that we use in Methods~\ref{App:JIT}. A unit cube as defined here contributes one unit of distance to the system in both the spatial and temporal directions.

We consider two different lattices that have been discussed in Methods~\ref{SubSec:Lattices}; the standard surface code, and the surface code on the alternative lattice that we show in Fig.~\ref{Fig:AltLattice}. Both lattices include qubits lying on the edges of a standard cubic lattice. There are twelve edges on the boundary of each unit cube but, as we see in Fig.~\ref{Fig:LatticeSurgeryPeriod}, the unit cubes are such that there are $\sim d\times d$ edges that are shared between two cubes, as well as $\sim d \times d $ faces, each consisting of four edges, that are shared between pairs of cubes. We therefore find seven qubits per unit cube lying on the edges of the cubic lattice.

We also assume that there is a single qubit for each plaquette measurement needed to produce the lattices shown in Fig.~\ref{Fig:LatticeSurgeryPeriod} at time $t + 1/2$. For the standard lattice surface code there are six plaquette measurements associate to each unit cube; one for each of its faces. However, as shown in Fig.~\ref{Fig:LatticeSurgeryPeriod} at time $t$, two of the faces have already been measured during an earlier cycle. Further, two face measurements of each unit cube are shared with other unit cubes, we therefore count three measurement ancilla qubits per unit cube for the standard surface code. In total, including the qubits on the edges of the lattice, we find 10 qubits per unit cell of the standard lattice surface code. A similar analysis finds that we need to perform four plaquette measurements per unit cube to produce a slice of the alternative surface code at time $t+1/2$. The alternative surface code thus includes 11 qubits per unit cell.

To conserve resources we assume that the two stationary qubit arrays support the two alternative lattice surface codes. Each of these arrays therefore requires $11 d^2 $ qubits to produce $d\times d $ unit cells. Similarly, the resource cost of $3d^2$ unit cells of the conventional cubic lattice surface code on the long array uses $10\cdot 3 d^2$ qubits. In total, all three arrays support $ \sim [30 + 2 \cdot 11] d^2 = 52 d^2  $ qubits. Assuming that the full protocol is completed in $3d$ cycles we arrive at a total resource cost of $156d^3$ physical qubit cycles for a single implementation of the gate.

The conservative estimate given above assumes that $10 \cdot 2d \times d$ qubits are idle for $3d$ units of time. We would obtain a resource saving of $60d^3$ qubit cycles by making use of these idle qubits or altering the protocol such that they are not needed. An easy way to achieve this is by simulating periodic boundary conditions on the long array. Indeed, we can achieve the same protocol by replacing the long array with a $d\times d$ array with cylindrical boundary conditions such that all three arrays have a size $\sim d\times d$ unit cells.

Periodic boundary conditions are easily achieved given a distributed architecture~\cite{Barrett05} where we are not constrained to strictly local interactions. One could also imagine approximating periodic boundary conditions with a strictly local system using a line of $L$ gates that share one very long array. The very long array has size $(L+2) d \times d$ and supports $L$ disjoint $d\times d $ surface codes. All $L$ gates proceed in parallel where all $L$ codes move synchronously along the very long array. In both cases, in the latter where $L$ diverges, we arrive at a resource cost of $\sim 96 d^3$ qubit cycles per controlled-controlled-phase operation. Over the course of the gate we must perform $\sim 3 d^3$ controlled-controlled phase gates.

At this stage, one might be willing to make speculations on how the resource cost of the gate proposed here compares with well-studied magic-state distillation protocols. Let us take a recent example~\cite{Gidney19} where a magic-state distillation protocol is proposed that occupies $12 d' \times 6 d'$ qubits over $5.5 d'$ cycles, giving a total resource cost that approaching $ \sim 400{d'}^3$ qubit cycles. We deliberately choose to quantify the qubit cycles of this example with units of ${d'}^3$ instead of $d^3$. This is because, without numerical simulations, we cannot accurately calculate how the failure rate of the gate presented in this work decays with $d$ as compared with $d'$.

Optimistically, we might hope that the logical failure rates of both protocols decay comparably in distance. In which case we might compare resources whereby $d \sim d'$ and we find that the gate presented here can outperform magic-state distillation using a small fraction of the resources. In practice, gauge fixing will introduce additional errors while the controlled-controlled-phase gate proceeds. In contrast, a magic-state distillation procotol that uses only logical Clifford operations will not experience gauge fixing errors. As such we should expect that $d  > d' $ to obtain comparable logical failure rates. Presently, little work has been done to calculate the logical failure rate of gates that make use of gauge fixing. The extent of this problem will be very sensitive to the error rate of the plaquette measurements. In principle, errors introduced by gauge fixing are of a different nature to errors introduced by the environment. As we have discussed in the main text, an appropriately chosen decoder might be able to mitigate the errors introduced by gauge fixing.

Another reason one should anticipate that we should choose $d > d'$ is because the application of noisy controlled-controlled-phase gates on the physical qubits will introduce additional errors to the system. Of course, the noise introduced by these entangling gates depends on the implementation of these gates. For the discussion here, it is simpler to remain agnostic about the physical implementation of the logical gate. Further work needs to be done to determine the magnitude of these sources of noise.

\subsection{A just-in-time decoder}
\label{App:JIT}

Here we prove that the non-Clifford operation will perform arbitrarily well as we scale the size of the system provided the physical error rate on the qubits is suitably low. We outline an error-correction procedure as we undergo the controlled-controlled-phase operation. The argument requires two main components. We require a just-in-time decoder that controls the spread of an error during the gauge fixing, and we show that the spread errors are sufficiently small that we can correct them at a later stage. We first show that we can decode a spread error model globally during post processing using a renormalisation-group decoder before arguing that the error model is justified by the just-in-time decoder.

\subsubsection{Notation and terminology}

We suppose a local error model acting on the qubits of the spacetime of the non-Clifford process. For suitably low error rate we can characterise the errors as occurring in small, local, well-separated regions~\cite{Grimmett99, Gacs86, Gray01}. The just-in-time gauge fixing decoder will spread this error. Given the spread is controlled, we can show that a global renormalisation-group decoder~\cite{Harrington04, Bravyi13b} will correct the errors that remain after the gauge-fixing process. Our argument can be regarded as an extension of the threshold theorem presented in Ref.~\cite{Bravyi13b}. As such, we will adopt several definitions and results presented in~\cite{Bravyi13b}. We will also keep our notation consistent with this work where possible.

We divide the system into sites; small local groups of qubits specified on a cubic lattice. We consider an independent and identically distributed error model where a Pauli error occurs on a site with probability $p_0$. We say that a site has experienced an error if one or more of the qubits has experienced an error. Given a constant number of qubits per site, $N$, then the likelihood a site experiences an error $p_0 = 1 - (1-\varepsilon)^N$ is constant where each qubit of the system experiences an error with constant probability $\varepsilon$. We consider a Pauli error $E$ drawn from the probability distribution described by the noise model. We will frequently abuse notation by using $ E $ to denote both a Pauli operator, and the set of sites that support $E$.

The syndrome of an error $E$ is denoted $ \sigma(E) $. It denotes the set of defects caused by $E$. We say that a subset of defects of a syndrome can be neutralised if a Pauli operator can be applied such that all of the defects are neutralised without adding any new defects. We may also say that any such subset of the syndrome is neutral.

Defects lie at locations, or sites, $u = (u_x, u_y,u_t)$ in $2+1$-dimensional spacetime. The separation between two sites is measured using the $\ell_\infty$ metric where the distance between sites $u $ and $v $, denoted $|u-v|$ is such that $|u-v| = \min(|u_x-v_x|,\, |u_y - v_y|,\,| u_t - v_t|)$. We will be interested in regions of spacetime that contain a collection of points $M$. The diameter of $M$ is equal to $\max_{u,v \in M} |u-v|$. We say that a subset of points $M$ is $r$-connected if and only if $M$ cannot be separated into two disjoint proper subsets separated by a distance more than $r$.The $\Delta$-neighbourhood is the subset of sites that lie up to a distance $\Delta$ from a region $\rho$ together with the sites enclosed within region $\rho$ itself. Given that we have a local model in spacetime, defects appear on sites within the one-neighbourhood of the sites of the error $E$. The following argument relies heavily on the notion of a chunk at a given length scale $Q$.

\begin{defn}[Chunk]
Let $E$ be a fixed error. A level-$0$ chunk is an error at a single site $u\in E$. A non-empty subset of $E$ is called a level-$n$ chunk if it is the disjoint union of two level-$(n-1)$ chunks with diameter $\le Q^n / 2$.
\end{defn}

We express errors in terms of their chunk decomposition. We define $E_n$ as the subset of sites that are members of a level-$n$ chunk such that
\begin{equation}
E = E_0 \supseteq E_1 \supseteq  \dots \supseteq E_m, \label{Eqn:ErrorDecomp}
\end{equation} 
where $m$ is the smallest integer such that $E_{m+1} =\emptyset$. We then define subsets $F_j = E_j \backslash E_{j+1}$ such that we can obtain the chunk decomposition of $E$, namely
\begin{equation}
E = F_0 \cup F_1 \cup \dots \cup F_m. \label{Eqn:Chunk}
\end{equation}
A level-$m$ error is defined by the smallest value of $m$ such that $E_{m+1} = \emptyset$.

Expressing an error in terms of its chunk decomposition enables Bravyi and Haah~\cite{Bravyi13b} to prove that a renormalisation group decoder will decode any level-$m$ error with a sufficiently large system. The proof relies on the following lemma
\begin{lem} \label{Lemma:Diameter}
Let $Q \ge 6$ and $M$ be any $Q^n$-connected component of $F_n$. Then $M$ has a diameter at most $Q^n$ and is separated from other errors $E_n \backslash M$ by a distance greater than $ Q^{n+1} / 3$.
\end{lem}
The proof is given in Ref.~\cite{Bravyi13b}, see proposition 7. We note also that all of the defects created by a $Q^n$-connected component of $F_n$ lying in the 1-neighbourhood of the connected component are neutral. With this result, it is then possible to show that a renormalisation group decoder that finds and neutralises neutral $2^p$-connected components at sequentially increasing length scales $p$, will successfully correct an error provided $Q^m $ is much smaller than the size of the system. A threshold is then obtained using results from percolation theory~\cite{Grimmett99, Gacs86, Gray01} that show that for sufficiently low error rate, the likelihood that a level-$m+1$ chunk will occur is vanishingly small. The renormalisation-group decoder is defined as follows.

\begin{defn}[Renormalisation-group decoder]
The renormalisation group decoder takes a syndrome $\sigma(E)$ as input and sequentially calls the level-$p$ error-correction subroutine $\textsc{Error Correct}(p)$ and applies the Pauli operator returned from the subroutine for $p = 0,\,1,\dots ,\, m$ with $m \sim \log L$. 
\end{defn}

The subroutine $\textsc{Error Correct}(p)$ returns correction operators for neutral $2^p$-connected subsets of the syndrome. If the syndrome has not been neutralised after $\textsc{Error Correct}(m)$ has been called the decoder reports failure.

\subsubsection{A threshold theorem with a spread error}
\label{SubSec:ThresholdTheorem}

In the following subsection we will show that the just-in-time gauge-fixing process will spread each disjoint $Q^j$-connected component of $F_j$ such that the linear size of the area it occupies will not increase by more than a constant factor $ s \ge 1 $. Once the error is spread during the gauge-fixing process we must show that the error remains correctable. Here we show that the spread error model can be corrected globally with the renormalisation-group decoder. We first define a level-$m$ spread error.

\begin{defn}[Spread errors]
Take a level-$m$ error $E$ drawn from an independent and identically distributed noise model with a chunk decomposition as in Eqn.~(\ref{Eqn:Chunk}). The spread error takes every $Q^j$-connected component $F_{j,\alpha} \subseteq F_j$ for all $j$, and spreads it such that this component of the error, together with the defects it produces, are supported within a container $C_{j,\alpha}$ centred at $F_{j,\alpha}$ with diameter at most $sQ^j$.
\end{defn}
We use the term `container' so we do not confuse them with boxes used in the following subsection, although containers and boxes both perform similar tasks in the proof.

In the proof given in Ref.~\cite{Bravyi13b} the authors make use of Lemma~\ref{Lemma:Diameter} to show that the renormalisation-group decoder will not introduce a logical failure. This is assured given that all of the errors are small and well separated in a way that is made precise by Lemma~\ref{Lemma:Diameter}. With the errors of the spread error model now supported in containers as much as a factor $s$ larger than the initial connected components of the error, the connected components are now much closer together, and in some cases overlap with one another. We have to check that the noise will not introduce a logical failure given sufficiently low noise parameters. We will argue that we can still find a threshold error rate provided $(s+2)Q^m$ is suitably small compared with the system size. The following definition will be helpful.

\begin{defn}[Tethered]
Consider errors supported within spread containers $C_{j, \alpha}$ and $C_{k,\beta}$ with $j \le k$. We say the error in container $C_{j,\alpha}$ is tethered to the error in a different container $C_{k, \beta}$ if the two containers are separated by a distance no greater than $\Delta_j$ where $\Delta_j = [r(s+2)+2]Q^j$. We say that $C_{j,\alpha}$ is untethered if it is not tethered to any containers $C_{k,\beta}$ for $k \ge j$. 
\end{defn}

We include an $r$ term to parameterise the separation we wish to maintain between untethered containers compared to the diameter of the containers. This should be of the order of the factor by which renormalisation-group decoder increases its search at each level. We defined the renormalisation-group decoder to search for $2^p$-connected components at level $p$, so we can take $r \ge 2$.

\begin{fact} 
Let $Q \ge 3[r(s+2) + s + 1]$. Two distinct containers of the same size, $C_{j,\alpha}$ and $C_{j,\beta}$, are not tethered. \label{Fact:Boxes}
\end{fact}
\begin{proof}
Errors $F_{j,\alpha}, \, F_{j,\beta} \subseteq F_j$ at the centre of spread errors contained in containers $C_{j,\alpha}$ and $C_{j, \beta}$ are separated by more than $Q^{j+1} / 3$; Lemma~\ref{Lemma:Diameter}. After expansion the boundaries of $C_{j,\alpha}$ and $C_{j,\beta}$ are separated by a distance greater than $Q^{j+1} / 3 - (s-1)Q^j$. We have $ \Delta_j \le Q^{j+1} / 3 - (s-1)Q^j $ for $Q \ge 3[r(s+2) + s+1]$. Therefore two boxes of the same size are not tethered for $Q \ge 3[r(s+2)+s+1]$.
\end{proof}

The constant expansion of the diameter of the errors means that some large errors expand such that smaller errors are not locally corrected. Instead they become tethered to the larger errors that may cause the renormalisation-group decoder to become confused. We will show that the small errors that are tethered to larger ones are dealt with at larger length scales as tethering remains close to the boundary of the larger containers with respect to the length scale of the larger container. We illustrate this idea in Fig.~\ref{Fig:NotToScale}.

\begin{figure}
\includegraphics{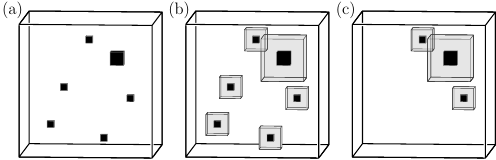}
\caption{\label{Fig:NotToScale} Not to scale. A diagram sketching the proof of a threshold for the controlled-controlled-phase gate. (a)~An error described by the chunk decomposition acting on the qubits included on the spacetime of the controlled-controlled-phase gate. See Lemma~\ref{Lemma:Diameter}. The image shows connected components of the error contained within black boxes. Errors are shown at two length-scales. One error at the larger length scale is shown to the top right of the image. (b)~After just-in-time gauge fixing is applied, errors are spread by a constant factor of the size of the connected components. This is shown by the grey regions around each of the initial black errors. (c)~Given a sufficiently large $Q$ the spread is not problematic since smaller untethered spread errors are far away from other components of equal or greater size. They are therefore easily dealt with by the renormalisation-group decoder. Small components of the error that lie close to a larger error will be neutralised with the larger error close to its boundary.}
\end{figure}

We will say that a decoder is successful if it returns a correction operator that is equivalent to the error operator up to an element of the stabilizer group. Given that the logical operators of the model of interest are supported on containers with diameter no smaller than $L$, we say that a decoder is successful if an error and its correction is supported on a collection of well separated containers where each container is smaller than $L / 3$. It will be helpful to define fattened containers $\tilde{C}_{j,\alpha}$ that enclose the $Q^j$-neighbourhood of $C_{j,\alpha}$. The fattened containers have diameter $D_j \le (s+2)Q^j$. We also define the correction operator $R(p) $ which is the product of the correction operators returned by $\textsc{Error Correct}(p)$ for all levels up to level $p$. We are now ready to proceed with the proof.

\begin{lem} \label{Lemma:RG}
Take $Q \ge 3[r(s+2) + s + 1]$. The renormalisation-group decoder will successfully decode a level-$m$ error with constant spread-factor $s \ge 1$ provided $ D_m < L/3 $.
\end{lem}

\begin{proof}
We follow the progression of the renormalisation-group decoder inductively to show the correction is supported on the union of containers $\tilde{C}_{j,\alpha}$. We will prove that the renormalisation-group decoder satisfies the following conditions at each level $p$.

\begin{enumerate}
\item The correction operator $R(p)$ returned at level $p$ is supported on the union of fattened containers $\tilde{C}_{j,\alpha}$. \label{RG:step1}

\item For the smallest integer $l \ge 0$ such that $Q^l > 2^p$, modulo stabilizers, the error $R(p)E$ is supported within a $Q^l$-neighbourhood of an error contained in a container $C_{k,\alpha}$ for any $k$ such that its diameter is at least $s Q^l $. \label{RG:step3}

\item The restriction of $E$ and the level-$p$ correction operator $R(p)$ is the same up to stabilizers on fattened containers $\tilde{C}_{j,\alpha}$ of diameter $D_j \le 2^p$ for untethered containers $C_{j,\alpha}$.  \label{RG:step2}

\end{enumerate}

We prove the case for $p=0$. By definition, errors are supported on containers $C_{j,\alpha}$, therefore 1-connected components of the syndrome contained within $C_{j,\alpha}$ are supported on $\tilde{C}_{j,\alpha} $. This verifies condition~\ref{RG:step1}. Condition~\ref{RG:step3} holds by definition as follows. Since $Q^1 > 1$, tethered containers $C_{0, \alpha}$ of size no greater than $s$ are separated from at least one container $C_{j,\beta}$ for $j \ge 1$ by a distance no more than $\Delta_0$, otherwise it is untethered. This verifies that all tethered containers $C_{0,\alpha}$ lie entirely within the $Q$-neighbourhood of some container $C_{j,\beta}$ since $ s + \Delta_0 \le Q $. The containers $C_{j,\beta}$ that tether the errors in containers $C_{0,\alpha}$ are necessarily such that $j > 0$ by Fact~\ref{Fact:Boxes}. This verifies Condition~\ref{RG:step3} as we have shown containers $C_{0,\alpha}$ are only tethered to containers with diameter at least $sQ$. Condition~\ref{RG:step2} is trivial for $p = 0$ since all containers have diameter larger than 1.

We now suppose that the above conditions are true for $p$ to show that the conditions hold at $p+1$. We consider $\textsc{Error Correct}(p+1)$. We are interested in containers $C_{j,\alpha}$ such that the diameter of its fattened counterpart is such that $2^p < D_j \le 2^{p+1}$. We first find the smallest integer $l$ such that $Q^l > 2^{p+1}$. Since $ D_j = (s+2)Q^j \le 2^{p+1}$ we have $l \ge j+1$. There are two possible outcomes depending on whether $C_{j,\alpha}$ is tethered or not. We deal with each case separately.

If $C_{j,\alpha}$ is tethered it lies at most $\Delta_j $ from another container $C_{k,\beta}$ of diameter $sQ^k$ with $k > j$ by Fact~\ref{Fact:Boxes}. Given $\tilde{C}_{j,\alpha}$ has diameter no greater than $D_j $, we find that the error supported on $\tilde{C}_{j,\alpha}$ is supported entirely within the $ ( D_j + \Delta_j )$-neighbourhood of $C_{k, \beta}$. Expanding this expression we have that $ D_j + \Delta_j \le Q^{j+1} $ for $ Q \ge [(s+2) + r(s+2) + 2] $. This confirms condition~\ref{RG:step3} for error correction at level $p+1$.

In the case $C_{j,\alpha}$ is untethered, the fattened container $\tilde{C}_{j,\alpha}$, which is $D_j$-connected, is separated from all other containers that support uncorrected errors $\tilde{C}_{k,\beta}$ with $D_k \ge D_j$ by a distance greater than $\Delta_j - 2Q^j = r(s+2)Q^j$ by the definition of an untethered container. Given that $ D_j > 2^p $ we have that $r ( s+2 ) Q^j > 2^{p+1}$ for $r = 2$ at the level-$(p+1)$ error-correction subroutine. Therefore $\textsc{Error Correct}(p+1)$ will not find any components of $E$ outside of the container $\tilde{C}_{l,\alpha}$. As such a correction will be returned entirely on $\tilde{C}_{l,\alpha}$, verifying condition~\ref{RG:step2}.

We finally consider the support of the correction operator. If the error is tethered, the correction returned for $C_{j,\alpha}$ lies on some container $\tilde{C}_{k,\beta}$ with $k > j$ to which it is tethered. In the case of untethered errors the correction for each connected component supported on $C_{j,\alpha}$, and the correction for the smaller components tethered to it, is supported on its respective container $\tilde{C}_{j,\alpha}$. This verifies condition~\ref{RG:step1}.
\end{proof}

The argument given above says that all errors are corrected on well-separated containers that are much smaller than the size of the system provided $ D_m < L/3$. Given that there are no level-$m+1$ errors, all of the errors supported on containers of size $D_m$ will be untethered, and therefore corrected at the largest length scale. Therefore we bound the failure probability by predicting the probability that an error of size $Q^{m+1}$ occurs. Ref.~\cite{Bravyi13b} gives a formula stating that the likelihood that a level-$m$ chunk occurs on an $L\times L \times L$ lattice is
\begin{equation}
p _ m \le L^3 (3Q)^{-6} ( 3Q p_0)^{2 ^m}.
\end{equation}
Demanding that $(s+2)Q^m < L / 3$ we find $m =  [ \log (L/3)  - \log(s+2)] / \log Q \approx \log L / \log Q$ we find the logical failure rate decays exponentially in $L$ provided $(3Q)^6 p_0 < 1$. This demonstrates a threshold for $p_0 < (3Q )^6$. Taking $Q = 87$ using $s = 8$ and $r = 2$, and we have that the number of qubits per site is $N = 120$ from~\ref{App:Microscopics}, we obtain a lower-bound on the threshold error rate of $\varepsilon \sim 6 \cdot 10^{-15}$.

\subsubsection{Just-in-time gauge fixing}
\label{SubSec:JustInTime}

We use a just-in-time decoder~\cite{Bombin18} to fix the gauge of each topological cluster state onto a copy of the surface code. We can deal with each of the three codes separately since the three codes are yet to interact. We suppose we draw an error from the independent and identically distributed noise model that acts on the spacetime that is represented by the sites of the topological cluster state, see Methods~\ref{SubSec:Lattices} for the definition of a site of the models of interest. Note that more than one defect can lie at a given site since each site supports several stabilizers. We also assume that the state of the two-dimensional surface code on the initial face is such that the plaquette operators are in their +1 eigenstate although small errors may have been introduced to the qubits on the primal qubits of the initial face of the system. We defined the initial face in the main text, see Fig.~\ref{Fig:Raussendorf}(b). We justify this assumption by showing how we fix the gauge of the two-dimensional input system in the following SubSection.

We briefly review the gauge fixing problem that we already summarised in the main text. Face measurements that we obtain by measuring the dual qubits of the topological cluster state return random outcomes. However, due to the constraints among the stabilizers, these random outcomes are constrained to form loops if the system does not experience noise. To fix the gauge of the system we need only find a Pauli operator that restores the plaquettes to their +1 eigenstate. This correction can be obtained trivially by finding a Pauli operator that will move the loops to any smooth boundary that is far away from the initial face. Indeed, because the plaquettes at this boundary are initialised in the $+1$ eigenstate, we cannot terminate loops here. However, any other boundary is suitable. With the two-dimensional setup we have it is perhaps a natural choice to move the loops towards the terminal face. Up to a stabilizer, this correction will fill the interior of the loop. Ensuring that the initial face is fixed means that the correction for the gauge-fixing process is unique. Otherwise, there can be two topologically distinct corrections from the gauge-fixing process that can lead to a logical fault.

In the case that errors occur when we measure the dual qubits, strings will appear in incorrect locations. Given that in the noiseless case the loops should be continuous, we can identify errors by finding the locations where strings terminate. We refer to the end point of a broken string as a defect. Defects appear in pairs at the two endpoints of a given string. Alternatively, single defects can be created at a smooth boundary. We attempt to fix the gauge where the errors occur by pairing local defects to close the loops, or we move single defects to smooth boundaries to correct them. We then correct the gauge according to the corrected loop. However, we cannot guarantee that we corrected these loops perfectly, and the operator we apply to fix the gauge will appear as an error. Up to stabilizers, the error we apply during the gauge-fixing procedure will be equivalent to an error that fills the interior of the error loop. These errors are problematic after the transversal non-Clifford gate is applied. However, provided these errors are sufficiently small, we can correct them at a later stage of the error-correction process.

Correcting broken loops becomes more difficult still when we only maintain a two-dimensional layer of the three-dimensional system as it will frequently be the case that a single defect will appear that should be paired to another that appears later in the spacetime but has not yet been realised. As such, we will propagate defects over time before we make a decision on how to pair it. This deferral will cause the loop to extend over the time direction of the system and this, in turn, will cause gauge-fixing errors to spread like the distance the defects are deferred. However, if we can make the decision to pair defects suitably quickly, we find that the errors we introduce during gauge fixing is not unmanageable. Here we propose a just-in-time decoder that we can prove will not extend the size of an error uncontrollably. We assume that the error model will respect the chunk decomposition described above, see Eqn.~(\ref{Eqn:Chunk}). We find that the just-in-time decoder will spread each error chunk by a constant factor of its initial size. We give some more notation to describe the error model before defining the just-in-time decoder and justifying that it will give rise to small errors at a suitably low error rate.

We remember that the chunk decomposition of the error $E = F_1 \cup F_2 \cup \dots \cup F_m$ is such that a $Q^j$-connected component of $F_j$ has diameter no greater than $Q^j$ and is separated from all other errors in $E_j$, see Eqn.~(\ref{Eqn:ErrorDecomp}), by more than $Q^{j+1} / 3$. We also define the syndrome of the error $\sigma(E)$, i.e., the defects that appear due to error $E$. We also have that the error, together with its syndrome, is contained in a box $B_{j,\alpha}$ of diameter at most $Q^j +2$ to include syndromes that lie at the boundary of a given site where $F_{j,\alpha}$ is the $Q_j$-connected component of $F_j$.

We denote defects, i.e., elements of $\sigma(E)$ with coordinates $u$ according to their site. A given defect has a time coordinate $u_t$, and a two-dimensional position coordinate $u_x$ in the three-dimensional spacetime. We denote the separation between two defects $u$ and $v$ in spacetime by $|u - v|$ according to the $\ell_\infty$ metric. We also denote their temporal(spatial) separation by $|u_t - v_t|$($|u_x - v_x|$). At a given time $t$ which progresses as we prepare more of the topological cluster state, we are only aware of all defects $u$ that have already been realised such that $u_t \le t$. We neutralise the defects of the syndrome once we arrive at a time where it becomes permissible to pair them, otherwise we defer their pairing to a later time. Deferral means leaving a defect in the current time slice of the spacetime by extending the string onto the current time without changing the spatial coordinate of the defect. When we decide to pair two defects, we join them by completing a loop along a direct path on the available live qubits. In both cases we fix the gauge according to the strings we have proposed with the correction or deferral. We are now ready to define the just-in-time decoder that will accurately correct pairs of defects given only knowledge about defects $u$ where $u_t \le  t$.

\begin{defn}[Just-in-time decoder] \label{Def:JIT}
The just-in-time decoder, $\textsc{Just In Time}(t)$, is applied at each time interval. It will neutralise pairs of defects $u$ and $v$ if and only if both defects have been deferred for a time $\delta t  \ge |u - v|$. It will pair a single defect $u$ to a smooth boundary only if $u$ has been deferred for a time equal to its separation from the boundary.
\end{defn}

The definition we give captures a broad range of just-in-time decoders that could be implemented a number of ways. We could, for instance, consider clustering decoders~\cite{Harrington04, Bravyi13b, Anwar14, Watson14}, greedy decoders~\cite{Wootton15a} or possibly more sophisticated decoders based on minimum-weight perfect matching~\cite{Edmonds65, Dennis02, Raussendorf06} to implement the decoder. Here we only offer a simple rule that we can use to demonstrate a threshold within the coarse-grained picture of the chunk decomposition. We also remark that we might be able to find better decoders that do not satisfy the conditions of the just-in-time decoder proposed here. We make no attempt to optimise this, the goal here is only to prove the existence of a threshold using the simplest possible terms.

Before we show that the just-in-time decoder will introduce a spread error with a constant spread factor $s$ we first consider how the decoder performs if we consider only a single $Q^j$-connected component of the error $F_{j,\alpha} \subseteq F_j$. We first consider the $Q^j$-connected component of the error well isolated in the bulk of the lattice, and then we consider how it is corrected close to the boundary.

\begin{fact} \label{Fact:IsolatedBox}
The correction of an isolated $Q^j$-connected component of the error, $F_{j,\alpha}$, that lies more than $2(Q^j+2)$ from the boundary, is supported on the $(Q^j+1)$-neighbourhood of $B_{j,\alpha}$. No defect will exist for a time longer than $\delta t \sim 2(Q^j+1)$.
\end{fact}

\begin{proof}
Consider two defects $u,\,v $ contained in $B_{j,\alpha}$ at extremal points. These defects have separation at most $Q^j+2$. Let us say that $| u_t - v_t | = Q^j+2$ with $u_t > v_t$. The defect $v$ will be deferred for a time $2(Q^j+2)$ before it is paired a distance $Q^j+1$ from $B_{j,\alpha}$ in the temporal direction. This correction is supported on the $(Q^j+1)$-neighbourhood of $B_{j,\alpha}$. All defects of this component of the error will be paired before it is permissible to pair them to the boundary.
\end{proof}

By this consideration we obtain a constant spread parameter $\sim 3$ for boxes in the bulk of the model. We next consider the correction close to a smooth boundary. We find this will have a larger spread parameter.

\begin{fact} \label{Fact:BoxAtBoundary}
The correction of an isolated $Q^j$-connected component of the error, $F_{j,\alpha}$, produced by the just-in-time decoder is supported on the $ 3 ( Q^j + 2 ) $-neighbourhood of $B_{j,\alpha}$, if $ B_{j,\alpha} $ lies within $ 2 (Q^j +2)  $ of a smooth boundary. All defects will be neutralised after a time at most $3(Q^j+2)$.
\end{fact}

\begin{proof}
A defect $u$ lies at most $3(Q^j+2)$ from the boundary In the worst case all defects will be paired to the boundary after a time at most $3(Q^j+2)$. Considering a defect at an extremal location then, the just-in-time decoder may defer the correction of a defect beyond $B_{j,\alpha}$ at most $3(Q^j+2)$ in the temporal direction.
\end{proof}

The above fact gives a spread factor $s \sim 7$ for boxes close to the boundary. To be more specific, we might upper bound the spread factor with $s = 8$. So far we have only considered how the just-in-time decoder deals with well-isolated $Q^j$-connected components of the error. In fact, we find that for sufficiently large $Q$ all errors are well isolated in  a more precise sense. This is captured by the following lemma. We find that, given that any defect supported on a box $B_{j,\alpha}$ will be paired with another defect in the same box or to a nearby smooth boundary after a time at most $3(Q^j+2)$. It will never be permissible to pair defects contained in different boxes before they are terminated. In effect, all boxes are transparent to one another. This justifies the spread error model used in the previous subsection.

\begin{lem}
\label{Lemma:JIT}
Take a chunk decomposition with $Q \ge 33$. The just-in-time decoder will pair all defects supported on $ B_{j,\alpha}$ within the $3(Q^j+2)$-neighbourhood of $B_{j,\alpha}$ to either another defect in $B_{j,\alpha}$ or to the boundary.
\end{lem}

\begin{proof}
By Facts~\ref{Fact:IsolatedBox} and~\ref{Fact:BoxAtBoundary} we have that all the defects of isolated boxes $B_{j,\alpha}$ are paired to another defect in $B_{j,\alpha}$ or to a nearby smooth boundary at most $2(Q^j+2)$ from $B_{j,\alpha}$ after a time no more than $3(Q^j+2)$.

We may worry that the just-in-time decoder may pair defects within disjoint boxes if they are too close together. We consider the permissibility of pairing $u$ contained within $B_{j,\alpha}$ to $v$ contained in $B_{k,\beta}$. For $Q \ge 33 $ such a pairing will never be permissible. We suppose that, without loss of generality, the diameter of $B_{j,\alpha}$ is less than or equal to the diameter of $B_{k,\beta}$.  Given that $B_{j,\alpha}$ is separated from $B_{k,\beta}$ by a distance greater than $Q^{j+1} /3 - 2$, it will not be permissible to pair $u$ with $v$ within the lifetime of $u$ before it is paired to a boundary or another defect in $B_{j,\alpha}$ provided $ 3(Q^j +2) \le Q^{j+1} / 3 - 2  $. This is satisfied for all $j \ge 0$ for $Q \ge 33$.
\end{proof}

This Lemma therefore justifies our spread factor $s = 8$ used in the previous Subsection.

\subsubsection{Gauge prefixing}

Finally, we assumed that we can reliably prepare the plaquette operators on the initial face of the two-dimensional surface code in their $+1$ eigenstate. We can tolerate small errors on the edges of the initialised surface code, but a single measurement error made on a plaquette can cause a critical error with the just-in-time decoder as it may never be paired with another defect. This will lead to a large error occurring during gauge fixing, see Fig.~\ref{Fig:Prefixing}(a). It is therefore important to identify any measurement errors on the face measurements of the initial face before the gauge fixing begins. We achieve this by prefixing the plaquettes of the initial face of the topological cluster state before the controlled-controlled-phase gate begins. We run the system over a time that scales with the code distance before we commence the controlled-controlled-phase gate procedure. In doing so we can identify measurement errors that may occur on the dual qubits of the topological cluster state using measurement data collected before we conduct the non-Clifford operation. Fig.~\ref{Fig:Prefixing}(b) shows the idea; the figure shows that measurement errors can be determined by looking at syndrome data on both sides of a plaquette. We need only look at one side, namely, the side of the initial face before just-in-time gauge fixing takes place.

\begin{figure}
\includegraphics{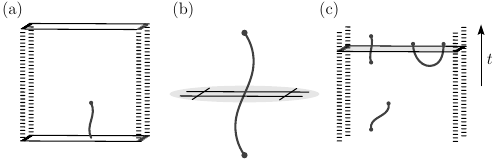}
\caption{(a)~A single measurement error on a face at the beginning of the controlled-controlled-phase gate will introduce a defect that may not be paired for a time that scales like the size of the system, this may introduce a macroscopic error after gaugefixing. (b)~We can determine where errors have occurred on the plaquettes of the initial face by looking at the defects before we begin the controlled-controlled-phase gate. (c)~We decode, or prefix the initial face, before we begin the controlled-controlled-phase gate to determine the locations of measurement errors on the initial face. \label{Fig:Prefixing}}
\end{figure}

Since we need only determine which face operators have experienced measurement errors and we do not need to actively correct the random gauge, gauge prefixing is accomplished globally using a renormalisation-group decoder on the three-dimensional syndrome data of the spacetime before the controlled-controlled-phase gate is performed.  A threshold can be proved by adapting the threshold theorem for topological codes given in Ref.~\cite{Bravyi13b}. Measurement errors close to the initial face before the controlled-controlled-phase gate takes place can then be identified easily by the decoder. We determine which plaquettes of the initial face have experienced errors by finding defects that should be paired to the initial face in the gauge prefixing operation. Small errors in the global gauge-prefixing procedure can be contained within the boxes that contain the error syndrome. As such, the errors that remain after the gauge prefixing procedure are confined within small boxes which respect the distribution we used to prove the threshold using the just-in-time decoder. As such we justify our error model used to bound the spread factor using just-in-time gauge fixing. We show an error together with its syndrome in Fig.~\ref{Fig:Prefixing}(c). The goal is only to estimate the plaquettes that have experienced measurement errors on the grey face at the top of the figure. This fixes the plaquettes of the initial face as we have assumed throughout our analysis.

We remark that the proposal given in Ref.~\cite{Bombin18} avoids the use of gauge prefixing by orienting boundaries such that the boundary that is analogous to the initial face of the color code is created over a long time. This orientation allows for single defects created at the initial face to be corrected by moving them back to the initial face at a later time, or onto some other suitable boundary. In contrast, here we have imagined that an initial face is produced at a single instant of time. Further work may show that we can apply the idea of Bomb\'{i}n to the surface code implementation of a controlled-controlled-phase gate presented here by reorienting the gate in spacetime. Such an adaptation will also require a modification of the just-in-time decoder to ensure that defects created at the initial face are paired to an appropriate boundary in a timely manner.

Conversely, gauge prefixing can be adapted for the proposal in Ref.~\cite{Bombin18}. In this work, color codes are entangled with a transversal controlled-phase gate. The transversal gate is applied to a two-dimensional support on boundaries of the two color codes undergoing this operation. Let us call this boundary the entangling boundary, where the initial face of the second code lies on the entangling boundary.  Let us briefly summarise how we can prefix the gauge of the initial face of the second of the two color codes by error correcting the first.

We note that the entangling operation allows us to use the eigenvalues of the error-detection measurements at the boundary of the first code to infer the values of the face operators at the initial face of the second code. Small errors may cause us to read the eigenvalues of the cells of the first code incorrectly. This will lead us to infer the wrong eigenvalues of the face operators of the initial face of the second code. However, error correction on the first code ensures that its entangling boundary is charge neutral, i.e., it has an even parity of string-like errors terminate at this boundary. If the first code is charge neutral at its entangling boundary, errors in the eigenvalues of the face operators of the initial face of the second color code are necessarily created in locally correctable configurations. This means they can be corrected without pairing any defects onto the initial face. This observation circumvents the need to orient the color code in a special configuration in spacetime. Relaxing this constraint may be of practical benefit. Moreover, the observation may allow us to remove certain rules that the decoder must otherwise respect to ensure defects are paired to the initial face as required. This may lead to improvements in the performance of the decoder.


\end{document}